\documentclass[12pt, draftclsnofoot, onecolumn,comsoc]{IEEEtran}
\usepackage[usenames,dvipsnames]{xcolor}
\usepackage[pdftex]{graphicx}
\graphicspath{{Figures/}}

\usepackage{cite}

\usepackage[utf8]{inputenc}
\usepackage[english]{babel}

\def\thspace{\hspace{0.25ex}}
\def\negthspace{\hspace{-0.25ex}}

\usepackage{amsthm}
\theoremstyle{remark}
\newtheorem{defn}{Definition}

\newtheorem{theorem}{Theorem}

\newtheorem{lemma}{Lemma}

\usepackage[cmex10]{amsmath}
\usepackage{amsthm,amssymb}
\usepackage{amsthm}

\usepackage{mathrsfs}
\newcommand{\Mathscr}[1]{\text{\usefont{U}{rsfs}{m}{n}#1}}
\newcommand{\Mathbb}[1]{\text{\usefont{U}{bbold}{m}{n}#1}}
\newcommand{\Mathcal}[1]{\text{\usefont{OMS}{cmsy}{m}{n}#1}}

\usepackage[scaled=0.92]{helvet} 

\def\T{\mathsf{T}}
\newcommand{\diff}[1]{\mathop{d\negthspace{#1}}}

\newcommand{\E}[1]{\Mathbb{E}[\thspace#1]}
\newcommand{\Var}[1]{\textup{Var}[\thspace #1]}
\newcommand{\Expectation}[1]{\Mathbb{E}\left[#1\right]}

\newcommand{\RevAdd}[1]{#1}
\newcommand{\RevMod}[1]{#1}
\usepackage{accents}
\newcommand{\ubar}[1]{\underaccent{\bar}{#1}}
\newcommand{\lah}[2]{\left\lfloor\begin{matrix} #1 \\ #2 \end{matrix}\right\rfloor}
\newcommand{\lahinline}[2]{\bigl\lfloor\!\begin{smallmatrix} #1 \\ #2 \end{smallmatrix}\!\bigl\rfloor}

\usepackage{array}
\newcolumntype{M}[1]{>{\centering\arraybackslash}m{#1}}

\usepackage[bookmarks,colorlinks,%
			pdfauthor={MPTL, GCF, TQSQ, MGDB},%
			pdftitle={TWC-Sparse-Fading}]{hyperref}%
\hypersetup{colorlinks,citecolor=gray,filecolor=blue!50,linkcolor=blue!50,urlcolor=blue!50}
\usepackage{enumerate}

\usepackage{caption}
\usepackage{booktabs}
\usepackage{xargs}
\usepackage{ifthen}

\newcommandx{\yaHelper}[2][1=\empty]{%
\ifthenelse{\equal{#1}{\empty}}%
  {\ensuremath{\scriptstyle{#2}}}
  {\raisebox{#1}[0pt][0pt]{\ensuremath{\scriptstyle{#2}}}}
}
\newcommandx{\yrightarrow}[4][1=\empty, 2=\empty, 4=\empty, usedefault=@]{%
  \ifthenelse{\equal{#2}{\empty}}%
  {\xrightarrow{\protect{\yaHelper[#4]{#3}}}}
  {\xrightarrow[\protect{\yaHelper[#2]{#1}}]{\protect{\yaHelper[#4]{#3}}}}
}
\def\toprinline{\,{\yrightarrow{\;\textup{P}\;}[-0.6ex]}\,}

\def\topr{\;{\yrightarrow{\ \textup{P}\ }[-0.6ex]}\;}

\usepackage{float}

\allowdisplaybreaks
\let\geq\geqslant
\let\leq\leqslant
\newlength\ploth 
\newlength\plotw

\usepackage{bm}

\begin{document}

\title{Fundamental Limits of Low-Density Spreading NOMA \RevAdd{with Fading}}
%
%
%

\author{Mai~T.~P.~Le,~\IEEEmembership{Student~Member,~IEEE},\\
        Guido~Carlo~Ferrante,~\IEEEmembership{Member,~IEEE},\\
        Tony Q.~S.~Quek,~\IEEEmembership{Senior Member,~IEEE},
       and~Maria-Gabriella~Di~Benedetto,~\IEEEmembership{Fellow,~IEEE}
\thanks{Mai~Thi~Phuong~Le is with the Department of Information Engineering, 
Electronics and Telecommunications, 
Sapienza University of Rome, Rome 00184, Italy (email: \texttt{mai.le.it@ieee.org})}
\thanks{Guido Carlo Ferrante was with the Singapore University of Technology and Design, Singapore 487372, and with the Massachusetts Institute of Technology, Cambridge, MA 02139 USA. He is now with Chalmers University of Technology, Gothenburg 41296, Sweden (email: \texttt{gcf@ieee.org})}
\thanks{Tony~Q.~S.~Quek is with the Information Systems Technology and Design Pillar, Singapore University of Technology and Design, 487372 Singapore (e-mail: \texttt{tonyquek@sutd.edu.sg})}
 \thanks{Maria-Gabriella~Di~Benedetto is with the Department of Information Engineering, 
Electronics and Telecommunications, 
Sapienza University of Rome, Rome 00184, Italy (email: \texttt{mariagabriella.dibenedetto@uniroma1.it})}

}

%
%

\markboth{Submitted to IEEE Transactions on Wireless Communications}%
{Shell \MakeLowercase{\textit{et al.}}: Bare Demo of IEEEtran.cls for IEEE Journals}
%



\maketitle

\begin{abstract}
Spectral efficiency of low-density spreading non-orthogonal multiple access channels in the presence of fading is derived for linear detection with independent decoding as well as optimum decoding. The large system limit, where both the number of users and number of signal dimensions grow with fixed ratio, called load, is considered. In the case of optimum decoding, it is found that low-density spreading underperforms dense spreading for all loads. Conversely, linear detection is characterized by different behaviors in the underloaded vs. overloaded regimes. In particular, it is shown that spectral efficiency changes smoothly as load increases. However, in the overloaded regime, the spectral efficiency of low-density spreading is higher than that of dense spreading.
\end{abstract}

\begin{IEEEkeywords}
Spectral efficiency, multiple access channels, non-orthogonal multiple access (NOMA). 
\end{IEEEkeywords}

%

\section{Introduction}

\subsection{Background and Motivation}
\label{background}
While expected to be standardized by the year 2020, the fifth generation (5G) currently receives considerable attention from the wireless community \cite{AndBuzChoetal:2014}. Massive \RevMod{multiple-input multiple-output (MIMO)}, millimeter-wave communications, ultra-dense networks, and non-orthogonal multiple access (NOMA) are four promising technologies, that are expected to address the targets of 5G wireless communications, including high spectral efficiency, massive connectivity, and low latency \cite{DaiWanYuaetal:2015, BocHeaLozetal:2014}. 

Back to the history of cellular communications from 1G to 4G, the radio multiple access schemes are mostly characterized by orthogonal multiple access (OMA), where different users are assigned to orthogonal resources in either frequency (\RevMod{frequency-division multiple access (FDMA) and orthorgonal FDMA (OFDMA)), time (time-division multiple access (TDMA)) or code (synchronous code-division multiple access (CDMA)} in underloaded condition) domains. However, 5G multiple access is required to support a wide range of use cases, providing access to massive numbers of low-power internet-of-thing (IoT), as well as broadband user terminals in the cellular network. Providing high spectral efficiency, while minimizing signaling and control overhead to improve efficiency, may not be feasible to achieve by OMA techniques \cite{Qualcomm:2015}. In fact, the orthogonality condition can be imposed as a requirement only when the system is \textit{underloaded}, that is, when the number of active users is lower than the number of available resource elements (degrees of freedom or dimensions).

The idea of NOMA is to serve multiple users in the same band and abandon any attempt to provide orthogonal access to different users as in conventional OMA. Orthogonality naturally drops when the number of active users is higher than the number of degrees of freedom, and ``collisions'' appear. One possible way of controlling collisions in NOMA is to share the same signal dimension among users and exploit power (\RevMod{power-domain NOMA (PDM-NOMA)}) vs. code (\RevMod{code-domain NOMA (CDM-NOMA)}) domains \cite{DaiWanYuaetal:2015}. In PDM-NOMA, it uses superposition coding, a well-known non-orthogonal scheme for downlink transmissions \cite{CovTho:2012}, 
and makes superposition decoding possible by allocating different levels of power to different users \cite{Choi:2016}. The ``near'' user, with a  higher channel gain, is typically assigned with less transmission power, which helps making successive interference cancellation (SIC) affordable at this user \cite{Van:2012}. In CDM-NOMA, it is characterized by different dialects, such as low-density spreading CDMA (LDS) \cite{HosWatTaf:2008, RazHosIm:2011, BeePop:2009}, low-density spreading \RevMod{orthogonal frequency-division multiplexing (LDS-OFDM)} \cite{RazImaIm:2012}, sparse code multiple access (SCMA)\cite{ZhaZhoZho:2017}, pattern division multiple access (PDMA)\cite{CheRenGaoetal:2017}, and multi-user shared access (MUSA)\cite{YuaYuLi:2016}. As a matter of fact, CDM-NOMA variants enable flexible resource allocation, and reduce hardware complexity by relaxing orthogonality requirements. 

In this work, we focus on LDS. \RevAdd{As a typical variant of CDM-NOMA, LDS inherits all above advantages and will be shown later in this paper to obtain increased system throughput compare to conventional CDMA, particularly in massive communications. LDS, therefore, may be appropriately fit to IoT scenario \cite{DaiWanYuaetal:2015} and is also considered as a potential candidate for uplink machine-type-communications (mMTC) \cite{DaiWanYuaetal:2015}.} Conventional direct-sequence CDMA (DS) is based on the spread spectrum technique, that uses spreading sequences to spread the signal over a given bandwidth. 
In traditional CDMA, signal dimensions, also known as chips (the terminology stemmed from the chip-rate of the sample), are all filled in with nonzero values, making the structure of DS be a form of ``dense spreading'' with nonzero values commonly binary or spherical \cite{VerSha:1999}. The idea of LDS is to use spreading sequences that are the sparse counterparts of the dense spreading sequences of conventional CDMA; a fraction only of the dimensions is filled with nonzero entries \cite{HosWatTaf:2008}. The same concept of LDS can be found in \cite{FerDiB:2015} within the framework of time hopping CDMA, where time hopping and chips are mapped to frequency hopping and subbands, respectively. Specifically, the analysis therein can be 
considered as a reference for LDS in terms of information theoretic bounds. 
 
On the other hand, the massive connectivity of 5G wireless communications is modeled by letting the number of devices to be much larger compared to the number of degrees of freedom. 
The behavior of DS with random spreading was analyzed in the large system limit, where the number of users and dimensions go to infinity with same scaling, in pioneering works of Tse and Hanly \cite{TseHan:1999}, Tse and Zeitouni \cite{TseZei:2000}, Verd\'u and Shamai \cite{VerSha:1999}, and Shamai and Verd\'u \cite{ShaVer:2001}. Subsequently, LDS was similarly analyzed in \cite{FerDiB:2015} in the case of a channel without fading. There has been no investigation of the effect of frequency-flat fading so far on the spectral efficiency of LDS. 

Therefore, the goal of this paper is to fill the gap by investigating LDS within the information theoretic framework considered in \cite{VerSha:1999,ShaVer:2001,FerDiB:2015} in the presence of frequency-flat fading. We analyze fundamental limits in the large system limit when the number of simultaneous transmissions becomes large with respect to the number of degrees of freedom.

\subsection{Other Related Work}
Based on the scaling between the number of users and number of degrees of freedom, other related works beyond those mentioned so far investigated either large-scale systems \cite{MonTse:2006, RaySaad:2007,YosTan:2006} or small-scale systems \cite{HosWatTaf:2008,RazHosIm:2011,BeePop:2009}. The two different regimes require asymptotic derivations (as the number of users and degrees of freedom grow with same scaling) and non-asymptotic derivations (for finite values of the number of users and degrees of freedom), respectively. The aforementioned literature is detailed as follows:

\subsubsection{Large-scale system}
Most of prior works \cite{RaySaad:2007, YosTan:2006} on LDS in the large system limit was derived by means of the replica method, which was first used for DS by Tanaka \cite{Tan:2002}. Since the replica method is not rigorous, Tanaka's capacity formula was verified (up to a given load, called spinodal, approximately equal to $\beta_{\textup{s}} \approx 1.49$) in the large system limit by Montanari and Tse in \cite{MonTse:2006}, where random spreading with sparse sequences was used in the proof, jointly with belief propagation detection. 
Adopting the replica method, Raymond and Saad in \cite{RaySaad:2007} and Yoshida and Tanaka in \cite{YosTan:2006} analyzed binary sparse CDMA in terms of spectral efficiency with different assumptions on the sparsity level (i.e., the number of nonzero entries) $N_{\textup{S}}$ of signatures (in particular, $N_{\textup{S}}$ is a deterministic finite value in \cite{RaySaad:2007}, whereas $N_{\textup{S}}$ is a Poissonian random variable in \cite{YosTan:2006}). 

\subsubsection{Small-scale system}
Recent investigations \cite{HosWatTaf:2008,RazHosIm:2011,BeePop:2009} analyzed LDS with finite values for the number of users and signal dimensions, in the \textit{overloaded} regime, where the number of users exceeds the number of dimensions. 
In \cite{HosWatTaf:2008}, each user spreads data over a small number of dimensions (e.g., $N_{\textup{S}}=3$) with other dimensions being zero padded. The resulting spreading sequence for each user is then interleaved such that the signature matrix from all $K$ users appears to be very sparse. The analysis focused on the bit error rate for different receiver structures. A comparison with different received powers was also described to address the near-far problem. Using the same framework proposed in \cite{HosWatTaf:2008}, an information theoretic analysis of LDS with fading was presented in \cite{RazHosIm:2011} for a bounded numbers of active users. In particular, the capacity region of time-varying fading LDS channel was analytically determined and tested by simulation, given different sparsity levels and different maximum number of users per dimension. 

\subsection{Approach \RevAdd{and Contribution}}
\RevMod{In this paper, we extend the information theoretic framework of time- and frequency-hopping CDMA considered in \cite{FerDiB:2015} for LDS in the presence of frequency-flat fading along the lines of \cite{ShaVer:2001}. In \cite{FerDiB:2015}, the reference channel is the additive white Gaussian noise (AWGN) channel: in order to apply some of the result derived in \cite{FerDiB:2015} in an IoT setting, it is mandatory to extend the analysis to channels with fading. 
We propose an information theoretic analysis where achievable spectral efficiency with different receiver structures is derived for the case of sparse signatures ($N_{\textup{S}}=1$), and compare our results to the spectral efficiency of direct-sequence (DS) CDMA, which represents the archetypal example of dense spreading ($N_{\textup{S}}=N$), under the same input constraints such as energy per symbol and bandwidth \cite{ShaVer:2001}. }

\RevAdd{
The major contributions of this paper are as follows:
\begin{itemize}
\item A rate achievable with linear detection is derived in Theorem \ref{theo:SUMF} in closed form. It is possible to show that sparse signaling outperforms dense signaling when the network is overloaded ($K>N$) and that the effect of fading is to slightly increase the achievable rate in this region.
\item The spectral efficiency with optimum detection is derived in Theorem \ref{theo:C_opt} in closed form. It is possible to show that dense signaling outperforms sparse signaling in this setup.
\item The spectral efficiency with optimum detection is derived by finding the limiting spectral distribution of a matrix ensemble that jointly describes spreading and fading: this is a mathematical result of independent interest. The combinatorial structure of the moments of such distribution is compared to the combinatorial interpretations available for the case of LDS and DS without fading.
\item The spectral efficiency with optimum detection in the large system limit also validates the decoupling principle in the CDMA literature, showing its equivalence to the average rate of a set of parallel channels. Intuitively, the multiuser low-density NOMA with optimum detection may be interpreted as a bank of channels, where each channel experiences an equivalent single-user channel.
\item The results provide an insight into the design of signaling in dense networks. As envisioned in the IoT setting, many simple transceivers will be part of large networks: results in this paper suggest that, in the uplink of such networks, sparse signaling can achieve a rate several times larger than that achievable via dense signaling. 
\end{itemize}
}


\subsection{Organization}
The paper is organized as follows. Section \ref{sec:model} introduces the reference model for LDS based on the general framework of traditional DS with the same energy and bandwidth constraints. The most important results in the literature relevant to our analysis are recalled in Section \ref{sec:previousresults}. Achievable spectral efficiency of LDS with linear and optimum receivers are presented in Section \ref{sec:mainsec}. Finally, conclusions based on the comparison of fundamental limits of LDS in 5G network are drawn in Section \ref{sec:conclusion}.
\medskip

\textit{Notation}. Expectation operator is denoted $\Mathbb{E}$. We denote $[N]$ the set of integers $\{ 1,2,\dotsc,N\}$. $\bm{e}_i^{n}$ with $i\in [n]$ stands for the $i$\textsuperscript{th} vector of the canonical basis of $\Mathbb{R}^n$. The $j$\textsuperscript{th} element of a vector $\bm{v}$ is denoted by $[\bm{v}]_j$. Kronecker delta is denoted by $\delta_{ij}$, hence $\delta_{ij}=1$ if $i=j$ and $\delta_{ij}=0$ otherwise. If the base is not explicited, $\log$ means natural logarithm. Complex conjugation and hermitian transposition are denoted by $^{\ast}$. Convergence in probability of a sequence of random variables $(X_n)_{n\geq 0}$ to $X$ is denoted by $X_n \toprinline X $.

\section{Reference model}
\label{sec:model}

The proposed reference model of a LDS system in the presence of frequency-flat fading follows the traditional discrete complex-valued CDMA model 
\begin{equation} \label{eq:maineq}
\bm{y}=\bm{SAb}+\bm{n},
\end{equation}
where: $\bm{y}\in\Mathbb{C}^N$ is the received signal; $\bm{b}=[b_1, \dotsc, b_K]^{\T}  \in \Mathbb{C}^K$ is the vector of symbols transmitted by the $K$ users; $\bm{S}=[\bm{s}_1,\dotsc,\bm{s}_K] \in \Mathbb{R}^{N \times K}$ is a random spreading matrix, column $i$ being the unit-norm spreading sequence of user $i$; $\bm{A} \in \Mathbb{C}^{K\times K}$ is a diagonal matrix of complex-valued fading coefficients $\textup{diag}(a_{1},\dotsc, a_{K})$; and $\bm{n} \in \Mathbb{C}^N$ is a circularly symmetric Gaussian vector with a zero mean and covariance $\Mathcal{N}_{0}\bm{I}$. Users transmit independent symbols and obey the power constraint $\Mathbb{E}[|\thspace b_k|^2] \leqslant \Mathcal{E}$ for all $k$, hence
\begin{equation}
\label{eq:Power_constraint}
\E{\bm{b}\bm{b}^\ast}=\Mathcal{E}\bm{I}.
\end{equation}
The load of the system is defined as the ratio between the number of users $K$ and the number of dimensions $N$, and is denoted by $\beta:=K/N$. Systems with $\beta<1$ and $\beta>1$ are referred to as underloaded and overloaded systems, respectively. 

Both LDS and DS systems can be modeled by \eqref{eq:maineq} with sparse and dense spreading matrix $\bm{S}$, respectively. In the simplest models, all elements of $\bm{S}$ are nonzero in DS, e.g. $s_{ki} \in \{\pm 1/\sqrt{N}\}$, while all but one element per column is nonzero in LDS, i.e. $\bm{s}_{k}\in\{\pm\bm{e}^{N}_{i}\}_{i=1,\dotsc,N}$. For the sake of clarity, we define rigorously below what we mean by sparse vector and sparse matrix. 
\begin{defn}[Sparse vector]
\label{defn:sparse_seq}
A vector $\bm{v}\in\Mathbb{R}^{N}$ is $N_\textup{S}$-sparse if the cardinality of the set of its nonzero elements is $N_\textup{S}$, i.e. $\|\bm{v}\|_{0} :=|\{v_i \ne 0\}_{i=1,\dotsc,N}|=N_\textup{S}$.
\end{defn}

\begin{defn}[Sparse matrix]
\label{defn}
A matrix $\bm{S}=[\bm{s}_1,\dotsc,\bm{s}_K]$ is $N_\textup{S}$-sparse if each column $\bm{s}_k$ is an $N_\textup{S}$-sparse vector.
\end{defn}

A reference model for time- and frequency-hopping CDMA was presented in \cite{FerDiB:2015} building on the seminal paper \cite{VerSha:1999}. 
The present work extends the model of \cite{FerDiB:2015} by introducing fading along the lines of \cite{ShaVer:2001}. \RevAdd{Notice that the assumption underpinning the fading model is that fading coefficients do not change over the whole signature, and more generally over the whole coherence block. This assumption may seem to clash with the pursued large system analysis since the latter requires increasingly large signatures. However, notice that the large system limit is only used to derive closed form expressions of performance of interest: It is well known that results derived in the large system limit are in fact very good approximations of performance of finite systems. The only important assumption is to keep the same load $\beta$ in the finite system and in the large system.}

In the following, we consider the very sparse scenario corresponding to sparse matrices with $1$-sparse column vectors. In this case, each spreading sequence $\bm{s}_k$ contains only one nonzero element, equal to either $+1$ or $-1$, with equal probability. Hence, the energy of the sequence is concentrated in just one nonzero pulse, while in DS, the energy is uniformly spread over all $N$ dimensions. 

System performance is measured by spectral efficiency $C$, defined as the total number of bits per dimension, that can be reliably transmitted \cite{FerDiB:2015, VerSha:1999, ShaVer:2001}. 
The per-symbol signal-to-noise ratio (SNR) is given by \cite{Ver:2002}
\begin{equation} 
\label{eq:EbN0_R}
\gamma	:= \frac{\frac{1}{K}\E{\|\bm{b}\|^{2}}}{\frac{1}{N}\E{\|\bm{n}\|^{2}}} 
		= \frac{N}{K}\cdot\frac{b}{N}\cdot\frac{\Mathcal{E}_\textup{b}}{\Mathcal{N}_0} 
		= \frac{1}{\beta}\cdot C\cdot \eta,
\end{equation}
where $C=b/N$ is expressed in bits per dimension, $b$ is the number of bits encoded in $\bm{b}$, $\E{\|\bm{b}\|^{2}}=b\Mathcal{E}_\textup{b}$, $\E{\|\bm{n}\|^{2}}=N\Mathcal{N}_0$, and $\eta:=\Mathcal{E}_\textup{b}/\Mathcal{N}_0$.  

\section{Previous Results}
\label{sec:previousresults}
In this section, we summarize the results in the literature that are most relevant to our analysis, namely spectral efficiency for LDS without fading and spectral efficiency of DS with and without fading.  

\subsubsection{Spectral efficiency in the absence of fading for LDS and DS}
The model in \eqref{eq:maineq} reduces to that in \cite{FerDiB:2015} when $\bm{A}=\bm{I}$ (no fading). Optimum decoding with LDS achieves the following spectral efficiency:
\begin{equation}
 C^{\textup{opt}}_{\textup{lds}}(\beta,\gamma) = \sum_{k\geq 0} \frac{\beta^{k}e^{-\beta}}{k!}\log_{2}(1+k\gamma)\ \textup{bits/s/Hz}.
\end{equation}
Spectral efficiency with DS is \cite{VerSha:1999,ShaVer:2001}
\begin{multline}
 C^{\textup{opt}}_{\textup{ds}}(\beta,\gamma) = \beta \log_{2}\Big(1+\gamma-\frac{1}{4}\Mathcal{F}(\gamma,\beta) \Big)
  + \log_{2}\Big(1+\beta\gamma-\frac{1}{4}\Mathcal{F}(\gamma,\beta) \Big) \\ -\frac{1}{4\log 2}\cdot\frac{\Mathcal{F}(\gamma,\beta)}{\gamma}\ \textup{bits/s/Hz},
\end{multline}
where
\begin{equation}
\Mathcal{F}(x,z) = \bigg(\!\sqrt{x(1+\sqrt{z})^2+1}-\sqrt{x(1-\sqrt{z})^2+1}\,\bigg)^{ 2}.
\end{equation}
Linear detectors, such as single-user matched filter (SUMF), zero-forcing (ZF), and minimum mean square error (MMSE), result in the same mutual information with LDS. An achievable spectral efficiency for these multiple access channels is $R_{\textup{lds}}^{\textup{sumf}}=\beta I(b_1;r_1|\bm{S})$ (b/s/Hz) with $\bm{b}$ Gaussian and $\bm{S}$ sparse, where $I(b_1;r_1|\bm{S})$ is the achievable rate (bits/symbol) of user $1$: 
\begin{multline}
R_{\textup{lds}}^{\textup{sumf}}(\beta,\gamma)=R^{\text{zf}}_{\textup{lds}}(\beta,\gamma)=R^{\text{mmse}}_{\textup{lds}}(\beta,\gamma) =\beta\sum_{k\geq 0}\dfrac{\beta^k e^{-\beta}}{k!}\log_2\left(1+\frac{\gamma}{k\gamma+1} \right)\ \textup{bits/s/Hz}.
\end{multline}
Differently from LDS, linear detectors with DS achieve different spectral efficiency. Among the above mentioned linear detectors, MMSE achieves the highest spectral efficiency, which is equal to \cite{VerSha:1999}
\begin{equation}
 C^{\textup{mmse}}_{\textup{ds}}(\beta,\gamma)= \beta \log\Big(1+\gamma-\frac{1}{4}\Mathcal{F}(\gamma,\beta) \Big).
\end{equation}
 
\subsubsection{Spectral efficiency in the presence of fading for DS}
In the presence of fading, spectral efficiency with optimum decoding is \cite{ShaVer:2001}
\begin{equation}
 C^{\textup{opt}}_{\textup{ds}}(\beta,\gamma) =  C^{\textup{mmse}}_{\textup{ds}}(\beta,\gamma) + \frac{\eta-1-\log\eta}{\log 2}
\end{equation}
where $\eta>0$ satisfies the following fixed point equation
\begin{equation}
\eta = 1 - \beta + \beta \,\Expectation{\frac{1}{1+\eta|\thspace a|^{2}\gamma}},
\end{equation}
and $C^{\textup{mmse}}_{\textup{ds}}(\beta,\gamma)$ is spectral efficiency with MMSE, given by
\begin{equation}
C^{\textup{mmse}}_{\textup{ds}}(\beta,\gamma) = \beta \,\E{\log_{2}(1+\gamma|\thspace a|^{2}\eta)}.
\end{equation}
It is noteworthy that fading increases spectral efficiency with MMSE at high load. The intuition provided in \cite[Section III-C]{ShaVer:2001} is that some user appears very low-powered at the receiver, thus the ``interference population,'' i.e., the number of effective interferers is reduced. A similar behavior is not observed with ZF, which removes all interference irrespective of power. This effect is called ``interference population control.'' 
 
\section{Spectral Efficiency of LDS with Frequency-Flat Fading}\label{sec:mainsec}
In this section, we derive spectral efficiency with a bank of single-user matched filters and independent decoding in Section~\ref{subsec:SUMF} and with optimum decoding in Section~\ref{subsec:optdec}.

\subsection{Single-User Matched Filter (SUMF)}
\label{subsec:SUMF}
The decision variable for user $1$ is 
\begin{equation}
\begin{aligned} \label{eq:y1}
r_1 & = \bm{s}_{\negthspace 1}^{\T} \bm{y}\\
&=\bm{s}_{\negthspace 1}^{\T}  \left(\, \sum_{k=1}^K\bm{s}_k a_k b_k \right)+\bm{s}_{\negthspace 1}^{\T} \bm{n}\\
&=a_1 b_1+\sum_{k=2}^K\bm{s}_{\negthspace 1}^{\T}  \bm{s}_k a_k b_k  +\bm{s}_{\negthspace 1}^{\T} \bm{n},
\end{aligned}
\end{equation}
where the last step follows from the signatures being unit norm. %
Assuming Gaussian coding, $b_{k}\sim\Mathcal{N}_{\Mathbb{C}}(0,\Mathcal{E})$, the conditional mutual information (bits/symbol) for user $1$ is 
\begin{align}
\label{eq:I_SUMF_gen}
\hspace{-1ex}I(r_1;b_1|\bm{S},\bm{A}) & = I(y_1;b_1|\rho_{12},\dotsc,\rho_{1K},a_1,\dotsc,a_K) \nonumber \\
&=\Mathbb{E}\left[\log_2\left(1\hspace{-0.25ex}+\hspace{-0.25ex}\frac{|\thspace a_1|^2\gamma}{1\hspace{-0.25ex}+\hspace{-0.25ex}\gamma\sum_{k=2}^{K}\rho_{1k}^2|\thspace a_k|^2}\right)\right],
\end{align}
where $\rho_{1k}:=\bm{s}_{\negthspace 1}^{\T} \bm{s}_k$ and the expectation is taken with respect to $\lbrace{\rho_{12},\dotsc,\rho_{1K}}\rbrace$ and $\lbrace a_1,\dotsc,a_K\rbrace$. %
The corresponding mutual information of the multiuser channel is 
\[ R_{\textup{lds}}^{\textup{sumf}}(\beta,\gamma):=\beta I(r_1;b_1|\bm{S},\bm{A}) \  \textup{bits/s/Hz}. \] 
In the following theorem we propose an explicit form of \eqref{eq:I_SUMF_gen} for $1$-sparse matrices (cf. Definition \ref{defn}).

\begin{theorem}
\label{theo:SUMF}
\RevMod{ Let $\bm{S} \in \Mathbb{R}^{N\times K}$ be a $1$-sparse spreading matrix. In the large system limit, the following rate is achievable with a bank of SUMF detectors:}
\begin{equation} \label{I_TH_final}
R_{\textup{lds}}^{\textup{sumf}}(\beta,\gamma)=\frac{\beta}{\log 2}\int_{0}^{1}  
 \dfrac {e^{-t \left(\beta +\frac{1}{1-t}\cdot\frac{1}{\gamma}\right)}}  {1-t} \diff{t}\ \textup{bits/s/Hz}.
\end{equation}
\end{theorem}

\begin{IEEEproof} See Appendix \ref{app:theo1}.\smallskip
\end{IEEEproof}

\begin{figure*}[t]
\centering
\includegraphics{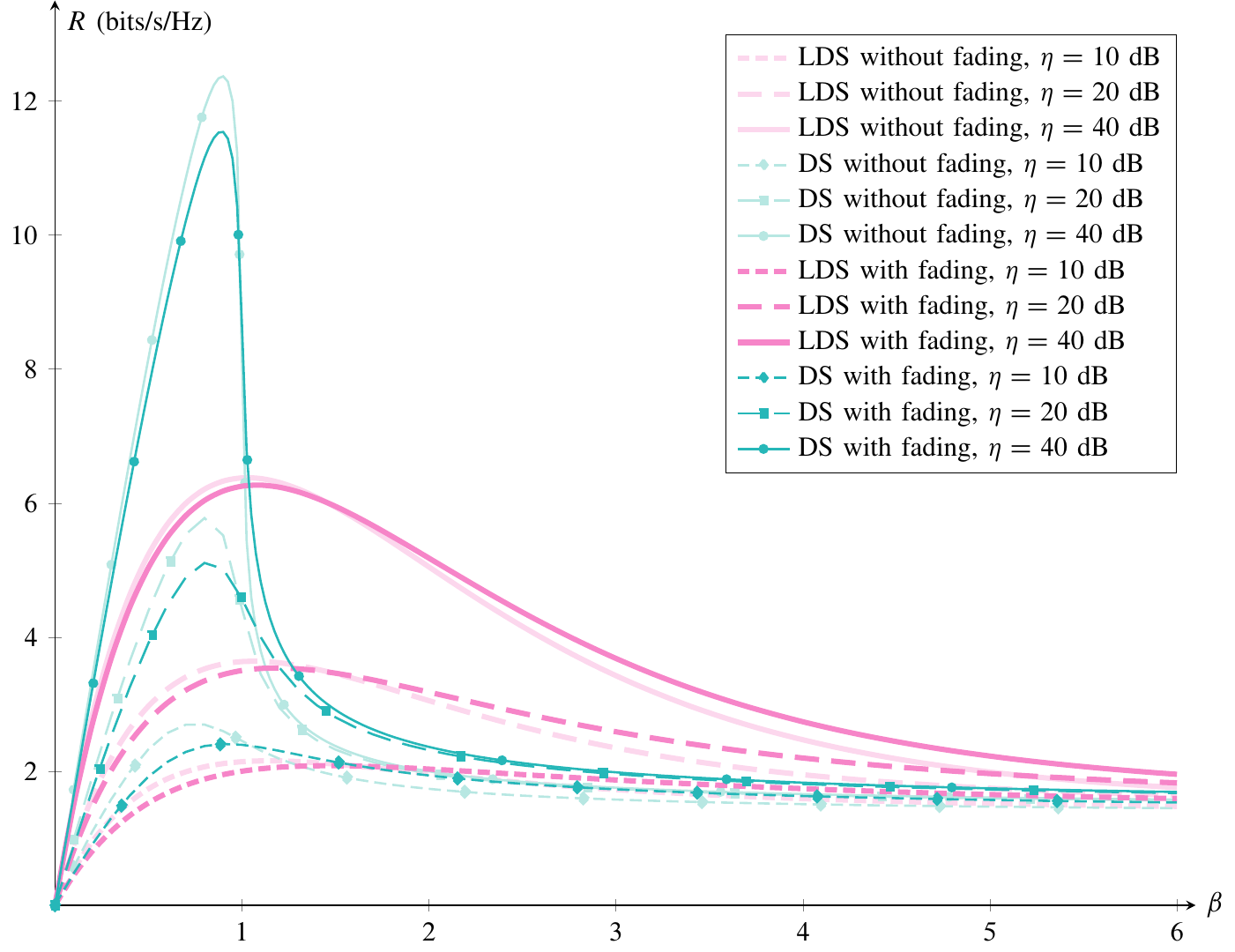}
\caption{Achievable spectral efficiency (bits/s/Hz) of LDS with SUMF detection (thick lines) and DS with MMSE detection (thin lines \RevAdd{with marks}) for several values of $\eta$ as a function of the load in the presence (dark shade) or absence (light shade) of fading.}
\label{fig:fig1}
\end{figure*}

The result in Theorem~\ref{theo:SUMF} allows us to study asymptotics for low and high SNR. In the low-SNR regime, the minimum energy per bit per noise level is given by (see Appendix~\ref{app:lowsnr} for the proof)
\begin{equation}\label{eq:etamin}
\eta_{\textup{min}} = \lim_{\gamma\to 0} \frac{\beta \gamma}{R_{\textup{lds}}^{\textup{sumf}}(\beta,\gamma)} = \log 2 \ \textup{dB},
\end{equation}
as in the case without fading, and with fading and dense spreading. Note that $\eqref{eq:etamin}$ holds for any $\beta>0$. The slope at $\eta=\eta_{\textup{min}}$ is (see Appendix~\ref{app:lowsnrslope} for the proof)
\begin{equation}\label{eq:slope0}
\hspace{-0.5ex}\Mathcal{S}_{0}^{\,\textup{sumf}} \hspace{-0.25ex}=\hspace{-0.25ex} 2\log 2\lim_{\gamma\to0}\frac{\big(\frac{\partial}{\partial\gamma}R_{\textup{lds}}^{\textup{sumf}}\thspace \big)^{ 2}}{-\frac{\partial^{2}}{\partial\gamma^{2}}R_{\textup{lds}}^{\textup{sumf}}} \hspace{-0.25ex}=\hspace{-0.25ex} \frac{\beta}{1+\beta}\ \textup{bits/s/Hz/(3 dB)},
\end{equation}
that is the same slope achieved with dense signaling. In the high-SNR regime, rate grows logarithmically with high-SNR slope equal to (see Appendix~\ref{app:slopeinf} for the proof)
\begin{equation}\label{eq:slopeinf}
\Mathcal{S}_{\infty}^{\,\textup{sumf}} = \log 2\lim_{\gamma\to\infty}\gamma \frac{\partial R_{\textup{lds}}^{\textup{sumf}}}{\partial\gamma} = \beta e^{-\beta}\ \textup{bits/s/Hz/(3 dB)},
\end{equation}
which is the same as LDS without fading, and, compared with the high-SNR slope achieved by DS with MMSE,
\begin{equation}\label{eq:slopeinfds}\Mathcal{S}_{\infty,\textup{ds}}^{\,\textup{mmse}}=\beta \Mathbb{1}_{\{\beta\in[0,1)\}} + \frac{1}{2} \Mathbb{1}_{\{\beta=1\}}+0  \Mathbb{1}_{\{\beta>1\}}, \end{equation}
shows that, for $\beta>1$, LDS is preferable to DS.

Figure \ref{fig:fig1} shows the achievable spectral efficiency with linear detection, and compares DS and LDS in the presence and absence of fading. In the presence of fading, the same qualitative phenomenon observed without fading holds, namely LDS outperforms DS, when load is approximately higher than unity. We stress that the curves for DS are capacities whereas the curves for LDS are merely achievable rates, and that this is sufficient to claim that LDS outperforms DS in the overloaded regime. The gap in performance with and without fading follows the same pattern for both DS and LDS, namely rates are decreased in the underloaded regime and increased in the overloaded regime. Finally, we notice that both LDS and DS are characterized by the same slope at $\beta=0$ and the same asymptotic value as $\beta\to\infty$. 

\subsection{Optimum decoding}
\label{subsec:optdec}
 
The spectral efficiency achieved with optimum decoding is the maximum (over the distributions on $\bm{b}$) normalized mutual information between $\bm{b}$ and $\bm{y}$ knowing $\bm{S}$ and $\bm{A}$, which is given by \cite{Ver:1986, ShaVer:2001}
\begin{equation}
\begin{aligned}
\label{eq:C_SAAS}
C_N^{\textup{opt}}(\beta,\gamma)
&=\dfrac{1}{N}\log_2 \det(\bm{I}+\gamma\bm{S}\bm{A}\bm{A}^{\ast}\bm{S}^{\ast}).
\end{aligned}
\end{equation}
We can express \eqref{eq:C_SAAS} in terms of the set of eigenvalues of the Gram matrix $\bm{S}\bm{A}\bm{A}^\ast\bm{S}^\ast$, $\{\lambda_{n}(\bm{S}\bm{A}\bm{A}^\ast\bm{S}^\ast)\colon 1 \leq n \leq N\}$, as follows:
\begin{equation}
\label{eq:C_opt_ESD}
C_N^{\textup{opt}}(\beta,\gamma)
=\int_0^{\infty} \log_2(1+\gamma\lambda) \diff{F_N^{\bm{S}\bm{A}\bm{A}^{\ast}\bm{S}^{\ast}}(\lambda)},
\end{equation}
being $F_N^{\bm{S}\bm{A}\bm{A}^{\ast}\bm{S}^{\ast}}(x)$ the \textit{empirical spectral distribution} (ESD) of $\bm{S}\bm{A}\bm{A}^{\ast}\bm{S}^{\ast}$, namely \cite{Gir:1990,BaiSil:2010}:
\begin{equation}
F_N^{\bm{S}\bm{A}\bm{A}^{\ast}\bm{S}^{\ast}}(x)
:= \frac{1}{N} \sum_{n=1}^N \Mathbb{1}_{\{\lambda_n(\bm{S}\bm{A}\bm{A}^{\ast}\bm{S}^{\ast}) \leq x\}}.
\end{equation}
Being $\bm{S}$ and $\bm{A}$ random, also $F_N^{\bm{S}\bm{A}\bm{A}^{\ast}\bm{S}^{\ast}}$ is random. In the large system limit, as is well known, the ESD can admit a limit (in probability or stronger sense), which is called \textit{limiting spectral distribution} (LSD) \cite{BaiSil:2010} and is denoted by $F(x)$. Hence, if the limit exists, spectral efficiency $C^{\textup{opt}}_N(\gamma)$ converges to
\begin{equation}
\label{eq:C_opt_LSD}
C^{\textup{opt}}(\beta,\gamma)=\int_0^{\infty} \log_2(1+\gamma\lambda) \diff{F(\lambda)}.
\end{equation}
Our main goal is, therefore, to find the LSD of the matrix ensemble $\{\bm{S}\bm{A}\bm{A}^{\ast}\bm{S}^{\ast}\}$. To this end, we compute in Theorem \ref{theo:mL} the average moments of the ESD in the large system limit and prove convergence in probability of the sequence of (random) moments of the ESD, 
\begin{equation}
\label{eq:m_L}
m_L:=\frac{1}{N}\textup{tr}(\bm{S}\bm{A}\bm{A}^{\ast}\bm{S}^{\ast})^L
=\int_0^{\infty}\lambda^L \diff{F_N^{\bm{S}\bm{A}\bm{A}^{\ast}\bm{S}^{\ast}}(\lambda)},
\end{equation}
to the (nonrandom) moments of the LSD. Then, by verifying Carleman's condition, Lemma \ref{lem:Carleman} shows that these moments uniquely specify the LSD \cite{Fel2:1968}. Finally, we use the LSD to derive the spectral efficiency in the large system limit in Theorem \ref{theo:C_opt}.

\begin{theorem}
\label{theo:mL}
Given the matrix ensemble $\{ \bm{S}\bm{A}\bm{A}^{\ast}\bm{S}^{\ast}\}$ with $\bm{S}$ an $N\times K$ sparse spreading matrix and $\bm{A}$ a $K\times K$ diagonal matrix of Rayleigh fading coefficients, it results
\begin{equation}
m_L \topr  \bar{m}_L:=\sum_{l=1}^L \lah{L}{l} \beta^l,
\end{equation}
where $\lahinline{L}{l}\,:=\binom{L-1}{l-1}\frac{L!}{l!}$ denotes the \textit{Lah numbers} \cite{Com:1974}.
\end{theorem}

\begin{proof} See Appendix \ref{app:mL}.
\end{proof}

In particular, $\bar{m}_L$ is the $L$\textsuperscript{th} moment of the random variable $\sum_{j=1}^J Z_j$ where $J$ is distributed according to a Poisson law with mean $\beta$ and, conditionally on $J$, $\{Z_j\colon 1 \leq j \leq J\}$ is a set of i.i.d. exponentially distributed random variables with unit rate. 

In the following lemma, we verify that the LSD is uniquely determined by the sequence of moments $(\bar{m}_L)_{L \geq 1}$.
\begin{lemma}
\label{lem:Carleman}
The sequence of moments $(\bar{m}_L)_{L\geq 1}$ satisfies the Carleman's condition, namely the series $\sum_{k\geq 1} \bar{m}_{2k}^{-1/(2k)}$ diverges.
\end{lemma}
\begin{proof} See Appendix \ref{app:Carleman}.
\end{proof}

\begin{figure*}[t]
\centering
\includegraphics{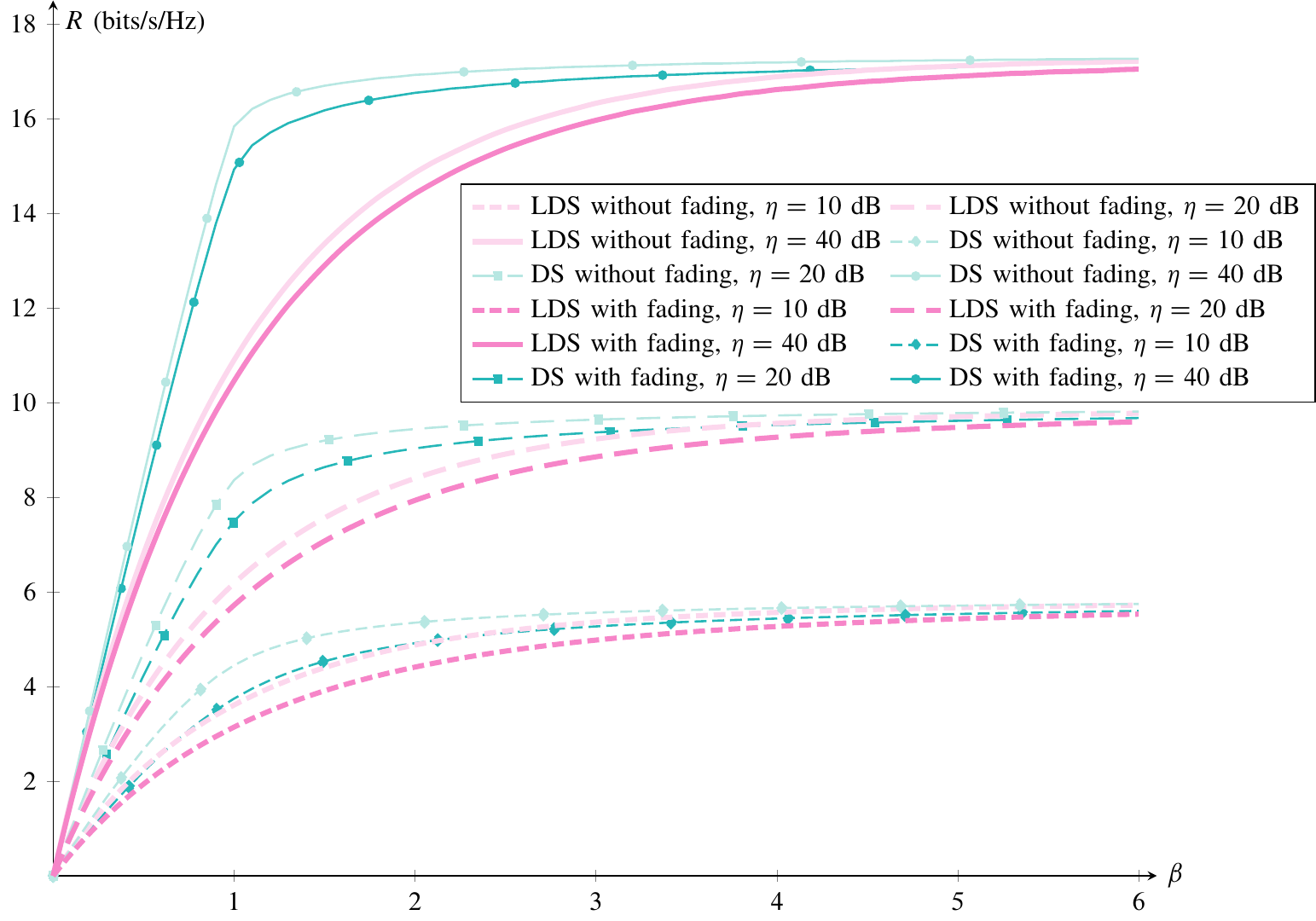}
\caption{Achievable spectral efficiency (bits/s/Hz) of LDS (thick lines) and DS (thin lines \RevAdd{with marks}) with optimum decoding for several values of $\eta$ as a function of the load in the presence (dark shade) or absence (light shade) of fading.}
\label{fig:fig1b}
\end{figure*}

Therefore, Theorem \ref{theo:mL} and Lemma \ref{lem:Carleman} imply that the probability measure $F(\lambda)$ in \eqref{eq:C_opt_LSD} is the probability measure of a compound Poisson distribution generated by the sum of a mean-$\beta$ Poissonian number of unit-rate exponentially distributed random variables:
\begin{equation}
\label{eq:meas}
F(\diff{\lambda}) = e^{-\beta} \delta_{0}(\diff{\lambda}) + \sum_{k\geq 1} \frac{e^{-\beta} \beta^{k}}{k!} \cdot \frac{e^{-\lambda} \lambda^{k-1}}{(k-1)!} \diff{\lambda}.
\end{equation}
The spectral efficiency in the large system limit is thus given by the average rate experienced through a set of parallel channels, indexed by $k=1,2,\cdots$, with signal-to-noise ratio equal to $\lambda\gamma$, used with probability $(e^{-\beta} \beta^{k}/k!) \cdot (e^{-\lambda} \lambda^{k-1}/(k-1)!) \diff{\lambda}$.
\RevAdd{ Indeed, this observation may validate the claim by Guo and Verd\'u that in the large system limit, the CDMA channel followed by multiuser detection can be decoupled into a bank of parallel Gaussian channels, each channel per user \cite{GuoVer:2005}. This is referred to as decoupling principle, which leads to the convergence of the mutual information of multiuser detection for each user to that of equivalent single-user Gaussian channel as the number of users go to infinitive, given the same input constraints. Given that the randomness of ESD vanishes in the large system limit (cf. Theorem \ref{theo:mL}), one may invoke the ``self-averaging'' property in the statistical physics \cite{GuoVer:2005}. Similarly to CDMA, the self-averaging principle yields to the strong property that for almost all realizations of the spreading sequences and noise of low-density NOMA, the macroscopic quantity (spectral efficiency in this case) converges to an equivalent deterministic quantity in the large system regime.
}

\begin{theorem}
\label{theo:C_opt}
The spectral efficiency with optimum decoding in the large system limit is given by
\begin{equation}
\label{eq:COPT}
\hspace{-0.75ex}C^{\textup{opt}}(\beta,\gamma)\hspace{-0.25ex}=\hspace{-0.25ex}\sum_{k\geq 1} \frac{e^{-\beta} \beta^{k}}{k!}\hspace{-0.25ex} \int_{0}^{\infty} \frac{e^{-\lambda} \lambda^{k-1}}{(k-1)!} \log_2(1+\gamma \lambda) \diff{\lambda}.
\end{equation}
\end{theorem} 

\begin{proof} Plug \eqref{eq:meas} into \eqref{eq:C_opt_LSD} and commute summation and integration, which is follows from Tonelli's theorem. 
\end{proof}

Similarly to the previous section, it is interesting also here to study the asymptotic behavior of spectral efficiency as a function of $\eta$. In the low-SNR regime, the minimum energy per bit per noise level is given by (see Appendix~\ref{app:etaminOPT} for the proof)
\begin{equation}\label{eq:etaminOPT}
\eta_{\textup{min}} = \lim_{\gamma\to 0} \frac{\beta \gamma}{C^{\textup{opt}}_{\textup{lds}}(\beta,\gamma)} = \log 2 \ \textup{dB},
\end{equation}
as in the case without fading, and with fading and dense spreading, irrespective of $\beta>0$. The slope at $\eta=\eta_{\textup{min}}$ is (see Appendix~\ref{app:slope0OPT} for the proof)
\begin{equation}\label{eq:slope0OPT}
\hspace{-0.5ex}\Mathcal{S}_{0}^{\,\textup{opt}} = 2\log 2\lim_{\gamma\to0}\frac{\big(\frac{\partial}{\partial\gamma}C^{\textup{opt}}_{\textup{lds}}\thspace\big)^{\negthspace 2}}{-\frac{\partial^{2}}{\partial\gamma^{2}}C^{\textup{opt}}_{\textup{lds}}}\hspace{-0.25ex} =\hspace{-0.25ex} \frac{2\beta}{\beta+2}\ \textup{bits/s/Hz/(3 dB)},
\end{equation}
which is the same as with dense signaling in the presence of fading (cf. (147) in \cite{ShaVer:2001}). In the high-SNR regime, rate grows logarithmically with high-SNR slope equal to (see Appendix~\ref{app:slopeinfOPT} for the proof)
\begin{equation}\label{eq:slopeinfOPT}
\Mathcal{S}_{\infty}^{\,\textup{opt}} = \log 2\lim_{\gamma\to\infty}\gamma \frac{\partial C^{\textup{opt}}_{\textup{lds}}}{\partial\gamma} = 1-e^{-\beta} \ \textup{bits/s/Hz/(3 dB)},
\end{equation}
which is the same as without fading.

\begin{figure*}[p]
\centering
\includegraphics{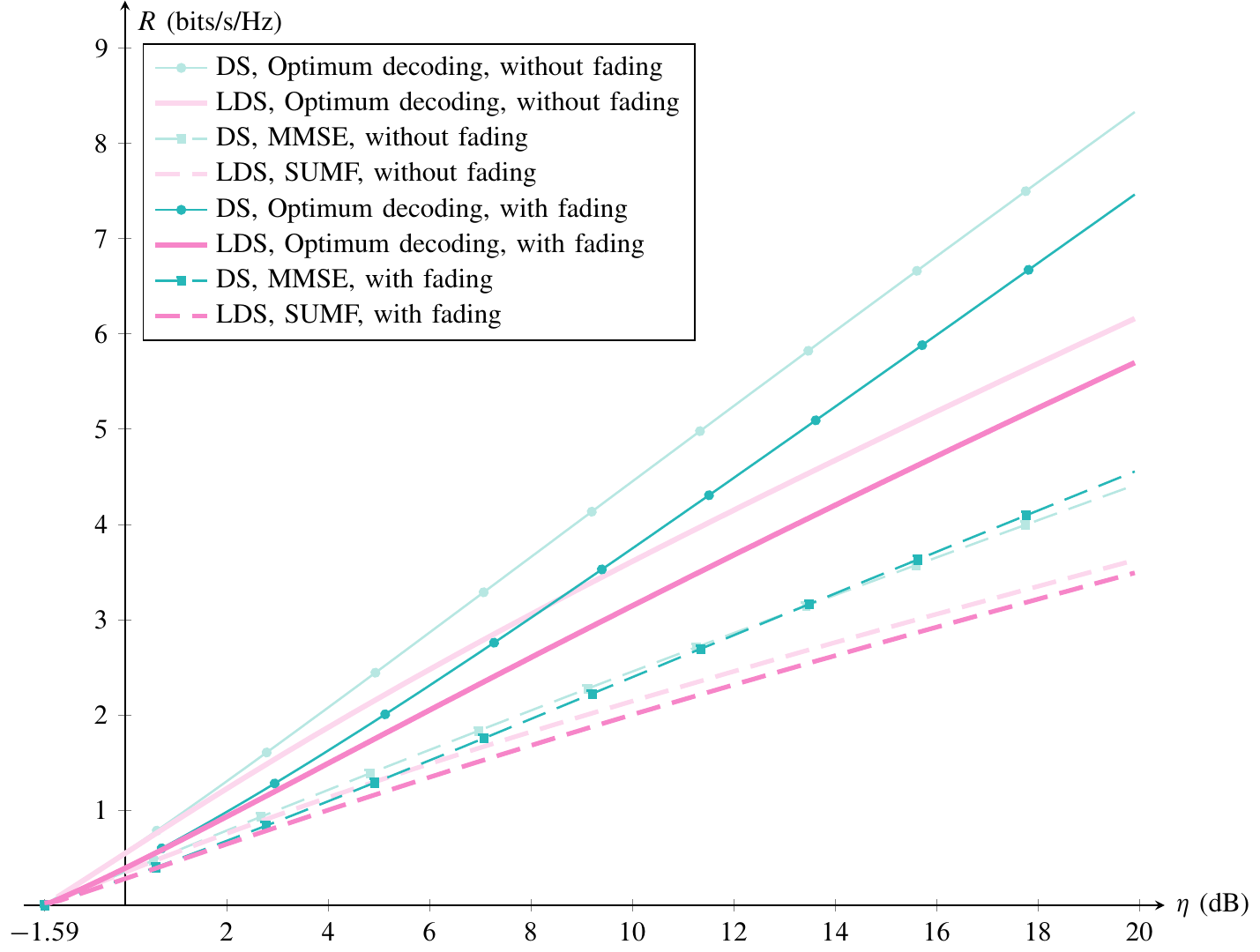}
\caption{Achievable spectral efficiency (bits/s/Hz) of LDS (thick lines) and DS (thin lines \RevAdd{with marks}) with optimum detection as a function of $\eta$ (dB) with load $\beta=1$ in the presence (dark shade) or absence (light shade) of fading. 
}
\label{fig:fig2}
\medskip
\includegraphics{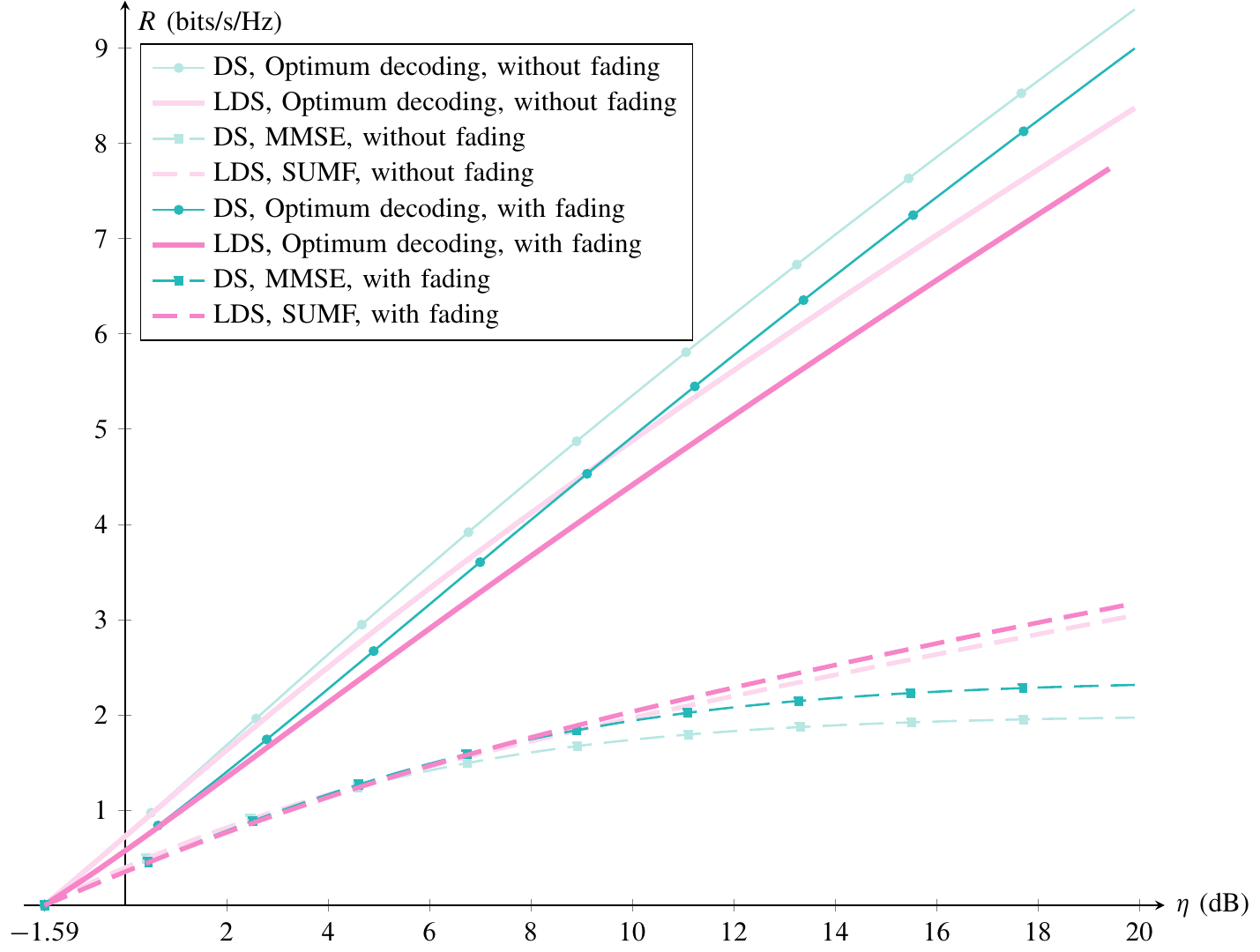}
\caption{Achievable spectral efficiency (bits/s/Hz) of LDS (thick lines) and DS (thin lines \RevAdd{with marks}) with optimum detection as a function of $\eta$ (dB) with load $\beta=2$ in the presence (dark shade) or absence (light shade) of fading. 
}
\label{fig:fig2b}
\end{figure*}


Figure \ref{fig:fig1b} shows the achievable spectral efficiency with optimum decoding and compares DS and LDS, in the presence and absence of fading. It is shown that, in general, LDS underperforms DS irrespective of fading; however, the main gap is concentrated around $\beta=1$, and decreases as load goes either to $0$ or $\infty$.
 
We conclude this section by highlighting the combinatorial connection between the moments found in Theorem \ref{theo:mL} and moments of the Mar\u{c}enko-Pastur and Poisson laws (see also Table~\ref{tab:tab1}), which correspond to the limiting spectral distributions of dense \cite{VerSha:1999} and sparse \cite{FerDiB:2015} schemes, respectively. We showed that the $L$\textsuperscript{th} moment is essentially a polynomial in $\beta$ with coefficients equal to Lah numbers. Similar results hold for dense and sparse schemes without fading, where Lah numbers are replaced by Narayana numbers and Stirling numbers of the second kind, respectively. All numbers are well-known in combinatorics: 
Narayana numbers enumerate non-crossing partitions into nonempty subsets; Stirling numbers of the second kind enumerate partitions into nonempty subsets; and Lah numbers enumerate partitions into nonempty linearly ordered subsets \cite{Com:1974}.

\begin{table*}[t]
\caption{Summary of LSDs, moments, and their combinatorial structure for different scenarios with and without fading.}
\label{tab:tab1}\small\bigskip
\centering
\begin{tabular}{c M{4cm} M{3cm} M{4.5cm} @{}m{0pt}@{}}
\toprule
\textbf{Scenario} & \textbf{LSD} & \textbf{LSD moment} $\bar{m}_{L}$ & \textbf{Coefficient}  &\\[10pt]
\midrule
DS with no fading & Mar\u{c}enko-Pastur law & $\sum_{l=1}^L \Mathcal{N}_l \beta^l$ & $\Mathcal{N}_l=\frac{1}{L} \binom{L}{l}\binom{L}{l-1}$ &\\[10pt] 
LDS with no fading & Poisson law & $\sum_{l=1}^L  {L \brace l}  \beta^l$ & ${L \brace l} = \frac{1}{l!} \sum_{j=0}^l (-1)^{l-j}\binom{l}{j}j^L$ &\\[10pt]
LDS with fading & Compound Poisson law &$\sum_{l=1}^L \lahinline{L}{l}\beta^l $ & $\lahinline{L}{l}= \binom{L-1}{l-1}\frac{L!}{l!}$ &\\[10pt]
\bottomrule
\end{tabular}
\medskip
\end{table*}

\subsection{Synopsis of results for LDS vs DS systems}
We collect the main results on LDS and DS systems from another perspective, namely for the case of fixed load and variable $\eta$, in Figs. \ref{fig:fig2} and \ref{fig:fig2b}. For fixed $\beta$, one can find the spectral efficiency as a function of $\eta$ by solving \eqref{eq:EbN0_R} with respect to $\gamma$ and computing the spectral efficiency for such value of $\gamma$. Achievable spectral efficiency (b/s/Hz), as a function of $\eta$ with optimum and linear detection in the presence and absence of fading, is shown for $\beta=1$ and $\beta=2$, respectively. To summarize the sources, results for DS were derived in \cite{VerSha:1999} without fading and in \cite{ShaVer:2001} with fading, whereas results for LDS without fading were derived in \cite{FerDiB:2015}.

Both figures show that all schemes are equivalent in the low-SNR regime, and that DS outperforms LDS with optimum decoding, particularly in the high-SNR regime, where spectral efficiency of the two schemes is characterized by different slopes. With linear detection, the scenario is completely different: when load increases beyond approximately unity, LDS outperforms DS, with a widening gap as $\eta$ increases. Indeed, LDS keeps a positive high-SNR slope while DS cannot afford it (cf. \eqref{eq:slopeinf} vs. \eqref{eq:slopeinfds}). Note the effect of the ``interference population control'' with DS (cf. Section~\ref{sec:previousresults}) on both figures: spectral efficiency with fading is higher than spectral efficiency without fading. A similar behavior holds with LDS as shown on Fig.~\ref{fig:fig2b}. 



\section{Conclusion}\label{sec:conclusion}
In this paper, a theoretical analysis of LDS systems in the presence of flat fading in terms of spectral efficiency with linear and optimum receivers was carried out in the large system limit, i.e., as both the number of users $K$ and the number of degrees of freedom $N$ grow unboundedly, with a finite ratio $\beta=K/N$. Spectral efficiency was derived as a function of the load $\beta$ and signal-to-noise ratio $\gamma$. The framework used extended the model in \cite{FerDiB:2015}, which was build on the seminal work \cite{VerSha:1999}, to the case with fading, along the lines of \cite{ShaVer:2001}. 

In the absence of fading, previous work showed that, in the large system limit DS has higher spectral efficiency than LDS when the system is underloaded ($\beta<1$). However, a drastic drop occurs at about $\beta=1$, and eventually, in the overloaded regime ($\beta >1$), LDS outperforms DS \cite{FerDiB:2015}. In this paper, we were to able to show that this is the case also in the presence of fading. 

This is particularly important in view of massive deployment of wireless devices and ultra-densification of the network towards 5G. Overloaded systems, where the number of resources is lower than the number of users accessing the network, will play a pivotal role in 5G, and this paper provides a theoretical ground for choosing LDS with respect to DS, and more generally choosing sparsity over density in signaling formats. 


Future investigations will focus on refining the understanding of overloaded systems with a more general structure of the sparsity of spreading sequences, e.g., when $N_{\textup{S}}>1$. It is interesting to understand which value of $N_{\textup{S}}$ represents the boundary between dense and low-dense systems, in terms of capacity, and more generally how the system behaves as a function of $N_{\textup{S}}$.


%
\appendices

\section{Proof of Theorem \ref{theo:SUMF}}
\label{app:theo1}
In order to compute \eqref{eq:I_SUMF_gen} we need to find the distribution of $\zeta:=\sum_{k=2}^{K}\rho_{1k}^2|\thspace a_k|^2$. We recall that $\rho_{1k}:=\bm{s}_{\negthspace 1}^{\T} \bm{s}_k$, hence the moment generating function (MGF) of $\rho_{1k}^{2}$ is
\begin{equation*}
M_{\rho^{2}}(t):=\E{e^{t\rho^{2}}} = \Big(1-\frac{1}{N}\Big)+\frac{1}{N}e^{t},
\end{equation*}
irrespective of $k$. Since $|\thspace a_k|^2$ is exponentially distributed with unit rate, we conclude
\begin{equation*}
M_{\rho^{2}|\thspace a|^2}(t):=\E{e^{t\rho^{2}|\thspace a|^2}} = \E{M_{\rho^{2}}(t|\thspace a|^{2})} = 1+ \frac{t}{(1-t)N}.
\end{equation*}
Therefore, the MGF of $\zeta$ is $(M_{\rho^{2}|\thspace a|^2}(t))^{K-1}$, and in the large system limit
\begin{equation*}
M_{\zeta}(t) \to e^{\beta\frac{t}{1-t}}.
\end{equation*}
Now, we express the logarithm via the following integral representation \cite{AbrSte:1972}
 \begin{equation*} 
\label{eq:log1x}
\log(1+x)= \int_{0}^{\infty} \frac{1}{s}(1-e^{-sx}) e^{-s} \diff{s},
\end{equation*}
which is valid for $x\geq 0$. Hence one has
\begin{equation*}
\log\left(1+\frac{|\thspace a|^2}{\zeta+1/\gamma}\right) = \int_{0}^{\infty} \frac{\diff{s}}{s} e^{-s/\gamma} (1-e^{-s|\thspace a|^{2}})e^{-s\zeta},
\end{equation*}
and by taking the expectation and changing variable, $t=\frac{s}{1+s}$,
\begin{equation*}
\Expectation{\log\left(1+\frac{|\thspace a|^2}{\zeta+1/\gamma}\right)}
=\int_{0}^{1} \frac {e^{-t \left(\beta +\frac{1}{(1-t)\gamma}\right)}}  {1-t} dt.
\end{equation*}

\section{Proof of \eqref{eq:etamin}}
\label{app:lowsnr}
With a change of variable, we can rewrite $R_{\textup{lds}}^{\textup{sumf}}(\beta,\gamma)$ as follows:
\[ R_{\textup{lds}}^{\textup{sumf}}(\beta,\gamma) = \frac{\beta}{\log 2}g_{\beta}(1/\gamma), \]
where 
\begin{equation}\label{eq:galpha}
g_{\beta}(\alpha):=\int_{0}^{\infty}\diff{z} e^{-\alpha z}e^{-\beta\frac{z}{1+z}}\frac{1}{1+z}. 
\end{equation}
We observe that $\eta_{\textup{min}}$ is expressed in terms of $g_{\beta}(\alpha)$ as follows:
\[ \eta_{\textup{min}} = \lim_{\alpha\to\infty} \frac{\log 2}{\alpha g_{\beta}(\alpha)},\]
hence we need to study $\alpha g_{\beta}(\alpha)$ as $\alpha\to\infty$. Since $g_{\beta}(\alpha)$ does not admit a closed form, we have to study the specific integral in \eqref{eq:galpha}. The basic observation is that the term $e^{-\beta\frac{z}{1+z}}$ is bounded on the integration interval from below and above, namely $e^{-\beta\frac{z}{1+z}}\in(e^{-\beta},1]$. Furthermore, most of the mass is concentrated in a neighborhood of $z=0$, as $\alpha$ increases. It makes sense to partition the domain $[0,\infty)$ in two subintervals, $[0,\epsilon)$ and $[\epsilon,\infty)$, for some $\epsilon>0$ fixed:
\begin{equation}\label{eq:galphasplit}
g_{\beta}(\alpha) =\int_{0}^{\epsilon}\diff{z} e^{-\alpha z}e^{-\beta\frac{z}{1+z}}\frac{1}{1+z}+\int_{\epsilon}^{\infty}\diff{z} e^{-\alpha z}e^{-\beta\frac{z}{1+z}}\frac{1}{1+z}.
\end{equation}
The first integral in \eqref{eq:galphasplit} is upper and lower bounded by
\begin{equation}\label{eq:galphasplit1}C_{1}(\epsilon)\int_{0}^{\epsilon}\diff{z} e^{-\alpha z}\frac{1}{1+z}, \end{equation}
with $C_{1}(\epsilon)=\bar{C}_{1}(\epsilon):=1$ and $C_{1}(\epsilon)=\ubar{C}_{1}(\epsilon):=e^{-\beta\frac{\epsilon}{1+\epsilon}}$, respectively. Similarly, the second integral in \eqref{eq:galphasplit} is upper and lower bounded by 
\begin{equation}\label{eq:galphasplit2} C_{2}(\epsilon)\int_{\epsilon}^{\infty}\diff{z} e^{-\alpha z}\frac{1}{1+z}, \end{equation}
with $C_{2}(\epsilon)=\bar{C}_{2}(\epsilon):=\ubar{C}_{1}(\epsilon)$ and $C_{2}(\epsilon)=\ubar{C}_{2}(\epsilon):=e^{-\beta}$, respectively. The integrals in \eqref{eq:galphasplit1}--\eqref{eq:galphasplit2} can be expressed by means of known functions,
\begin{align*}
\alpha\int_{0}^{\epsilon}\diff{z} e^{-\alpha z}\frac{1}{1+z}
	& =\alpha e^{\alpha}[E_{1}(\alpha)-E_{1}(\alpha(1+\epsilon))] \\
	& =1+O(1/\alpha),
\end{align*}
where $E_{1}(x)$ denotes the exponential integral\footnote{$E_{n}(x):=\int_{1}^{\infty}\diff{t}\frac{1}{t^{n}}e^{-xt}$ for all $x>0$ and $n$ positive integer.} for $x>0$, for which the following asymptotic expansion holds $\alpha e^{-\alpha}E_{1}(\alpha)=1-\alpha^{-1}+O(\alpha^{-2})$, and 
\begin{align*}
\alpha\int_{\epsilon}^{\infty}\diff{z} e^{-\alpha z}\frac{1}{1+z}
	& =\alpha e^{\alpha}E_{1}(\alpha(1+\epsilon)) \\
	& \leq e^{-\epsilon\alpha} (1+\epsilon)^{-1},
\end{align*}
which vanishes as $\alpha\to\infty$, where the inequality follows from the standard bracketing of $E_{1}$ through elementary functions. Therefore, we proved that, for all $\epsilon>0$, the term
\begin{equation}\label{eq:galphasplit1b} C_{1}(\epsilon)\;\alpha\int_{0}^{\epsilon}\diff{z} e^{-\alpha z}\frac{1}{1+z} = C_{1}(\epsilon) + O(1/\alpha),\end{equation}
contributes finitely to the integral, while the term
\begin{equation}\label{eq:galphasplit2b} C_{2}(\epsilon)\;\alpha \int_{\epsilon}^{\infty}\diff{z} e^{-\alpha z}\frac{1}{1+z} \to 0, \end{equation}
asymptotically vanishes. Hence, $\alpha g_{\beta}(\alpha)\to C_{1}(\epsilon)$ as $\alpha\to\infty$, and the result follows since $\epsilon>0$ is arbitrary.

\section{Proof of \eqref{eq:slope0}}
\label{app:lowsnrslope}
A sketch of the proof is provided. The slope can be written as follows:
\begin{equation}\label{eq:slope0bis}
\Mathcal{S}_{0} = \beta\lim_{\alpha\to\infty}\frac{\alpha I_{1}(\alpha)^{2}}{I_{1}(\alpha)-\frac{\alpha}{2}I_{2}(\alpha)},
\end{equation}
where 
\[ I_{k}(\alpha):=\int_{0}^{\infty} \diff{z} \frac{z^{k}}{1+z}e^{-z\big(\alpha+\frac{\beta}{1+z}\big)}. \]
As $\alpha\to\infty$, the mass of the integral is increasingly concentrated in a neighborhood of the origin, say $z\in[0,\epsilon]$:
\[ I_{k}(\alpha) \sim \int_{0}^{\epsilon} \diff{z} \frac{z^{k}}{1+z}e^{-z\big(\alpha+\frac{\beta}{1+z}\big)},\ \alpha\to\infty. \]
For any fixed $\epsilon>0$, it results $\frac{\beta}{1+z}\in[\frac{\beta}{1+\epsilon},\beta]$; therefore, as $\alpha\to\infty$ it also results
\[ I_{k}(\alpha) \sim \int_{0}^{\epsilon} \diff{z} \frac{z^{k}}{1+z}e^{-z(\alpha+\beta)},\ \alpha\to\infty. \]
The above integral can be expressed in closed form for $k=1$ and $k=2$. The result follows by computing the limit of the ratio in \eqref{eq:slope0bis}, which turns out not to depend on $\epsilon$.

\section{Proof of \eqref{eq:slopeinf}}
\label{app:slopeinf}
By explicitly computing the derivative in the definition of $\Mathcal{S}_{\infty}$, we can rewrite it as follows:
\begin{equation}\label{eq:slopeinf2}
\Mathcal{S}_{\infty} = \beta \lim_{\alpha\to0} \alpha I_{\beta}(\alpha),
\end{equation}
where 
\begin{equation}\label{eq:slopeinf3}
I_{\beta}(\alpha) := \int_{0}^{\infty} \diff{z} e^{-\alpha z}e^{-\beta \frac{z}{1+z}}\frac{z}{1+z}.
\end{equation}
The idea of the proof is to find upper and lower bounds on $\alpha I_{\beta}(\alpha)$, that match in the limit. To this end, observe that
\begin{equation}\label{eq:slopeinf4}
e^{-\beta}\frac{z}{1+z} \leq e^{-\beta \frac{z}{1+z}}\frac{z}{1+z} \leq e^{-\beta}\bigg(1+\frac{(\beta-1)^{+}}{z}\bigg).
\end{equation}
Hence, a lower bound is
\begin{align}
\alpha I_{\beta}(\alpha) 	& \geq \alpha \int_{0}^{\infty} \diff{z} e^{-\alpha z} e^{-\beta}\frac{z}{1+z} \nonumber \\
			& = e^{-\beta}(1-\alpha e^{\alpha}E_{1}(\alpha)) \nonumber \\
			& \to e^{-\beta}. \label{eq:slopeinf5}
\end{align}
In order to compute the upper bound, split the domain of integration to avoid a singularity at $z=0$ (cf. \eqref{eq:slopeinf4}) as follows. For any $\epsilon>0$, it results
\begin{align}
\alpha I_{\beta}(\alpha) 	& \leq \alpha \int_{0}^{\epsilon} \diff{z} z + \alpha \int_{\epsilon}^{\infty} \diff{z} e^{-\alpha z} e^{-\beta}\bigg(1+\frac{(\beta-1)^{+}}{z}\bigg) \nonumber \\
			& = \alpha \frac{\epsilon^{2}}{2} + e^{-\beta-\epsilon\alpha}+(\beta-1)^{+}\alpha E_{1}(\epsilon\alpha) \nonumber \\
			& \to e^{-\beta}, \label{eq:slopeinf6}
\end{align}
where for $z\in[0,\epsilon]$ we used a trivial upper bound for the integrand of $I_{\beta}(\alpha)$. The result follows from \eqref{eq:slopeinf2}, \eqref{eq:slopeinf5} and \eqref{eq:slopeinf6}.

\section{Proof of Theorem \ref{theo:mL}}
\label{app:mL}
In this appendix, we compute the moments \eqref{eq:m_L} and prove that convergence in probability to their mean holds.

The first remark is that matrix $\bm{S}\bm{A}\bm{A}^{\ast}\bm{S}^{\ast}$ is diagonal:
\begin{equation}
\label{eq:move_randomness}
\begin{aligned}
\bm{S}\bm{A}\bm{A}^{\ast}\bm{S}^{\ast}
&= \sum_{j=1}^K\bm{s}_j\bm{a}_j \bm{a}_j^ \ast \bm{s}_j^{\ast} \\
&= \sum_{j=1}^K  |\thspace a_j|^2 \bm{e}_{\pi_j}^{K} \bm{e}_{\pi_j}^{K\ast}\\
&=\sum_{i=1}^N\Bigg(\sum_{j=1}^K\Mathbb{1}_{\{\pi_j=i\}} |\thspace a_j|^2  \Bigg) \bm{e}_{i}^{N} \bm{e}_{i}^{N\ast},\\
\end{aligned}
\end{equation}
where $\pi_k$ denotes the nonzero element of the signature $\bm{s}_k$ and $\bm{e}_i^{n}$ denotes the $i$\textsuperscript{th} vector of the canonical basis of $\Mathbb{R}^{n}$. In the last step, we move randomness from vectors to scalars, which will be shortly useful. Indeed, $m_{L}$ can be written as follows:
\begin{equation}
\label{eq:m_L2}
\begin{aligned}
m_L &=\frac{1}{N} \textup{tr}(\bm{S}\bm{A}\bm{A}^{\ast}\bm{S}^{\ast})^L\\
&=\frac{1}{N} \sum_{i=1}^N ([\bm{S}\bm{A}\bm{A}^{\ast}\bm{S}^{\ast}]_{ii})^L\\
&= \frac{1}{N}  \sum_{i=1}^N \Bigg(\sum_{j=1}^K \Mathbb{1}_{\{\pi_j=i\}} |\thspace a_j|^2 \Bigg)^{\! L}.\\
\end{aligned}
\end{equation}
Call the sum in parenthesis $S_{i}$:
\begin{equation}\label{eq:Si}
S_{i}:= \sum_{j=1}^K \Mathbb{1}_{\{\pi_j=i\}} |\thspace a_j|^2.
\end{equation}
Hence, the expected value of $m_{L}$ is
\begin{equation}
\E{m_{L}} = \E{S_{1}^{L}}=M_{S_{1}}^{(L)}(0),
\end{equation}
where $M_{S_{1}}^{(L)}$ denotes the $L$\textsuperscript{th} derivative of the MGF of $S_{1}$. It can be shown that $M_{S_i}(t) = \big(1+ \frac{t}{(1-t)N}\big)^{\negthspace K}$, hence
\begin{equation}
M_{S_i}^{(L)}(0)=\sum_{l=1}^L \binom{L-1}{l-1} \frac{L!}{l!} \frac{K!}{(K-l)!\,N^l},
\end{equation}
which in the large system limit becomes
\begin{equation}
\label{eq:E_Y_L}
\E{m_{L}} \to \sum_{l=1}^L \lah{L}{l} \beta^l,
\end{equation}
where Lah numbers make their appearance $\lahinline{L}{l}:=\binom{L-1}{l-1} \frac{L!}{l!}$. Alternatively, $\E{m_{L}}$ can be expressed by using generalized Laguerre polynomials, which naturally appear in the Taylor expansion of the asymptotic MGF of $S_{1}$.

In order to prove convergence in probability, it is sufficient to show that $\Var{m_{L}}=\E{m_{L}^{2}}-(\E{m_{L}})^{2}\to 0$. We have already found $\E{m_{L}}$.  By using \eqref{eq:m_L2} and \eqref{eq:Si}, $\E{m_{L}^{2}}$ can be expressed as follows:
\begin{align}
\E{m_{L}^{2}} 
 & = \Mathbb{E}\thspace\bigg[{\bigg( \frac{1}{N}\sum_{i=1}^{N}S_{i}^{L} \bigg)^{\! 2}}\bigg]\\
 & = \frac{1}{N^{2}}\sum_{i=1}^{N} \E{S_{i}^{2L}} + \frac{1}{N^{2}}\sum_{i\neq j} \E{S_{i}^{L}S_{j}^{L}}.
\end{align}
The first term is $O(1/N)$ because $\E{S_{i}^{2L}}$ is bounded in the large system limit. The second term tends to $\E{S_{1}^{L}S_{2}^{L}}$. In order to show that this term becomes asymptotically equal to $\E{m_{L}}^{2}$, we can actually show more, namely $S_{1}$ and $S_{2}$ are asymptotically independent ($S_{1}\perp S_{2}$). 

To this end, interpret $S_{i}$ as the sum of a (random) number of weights $w_{k}:=|\thspace a_{k}|^{2}$, namely $S_{i}=\sum_{k\in\Mathscr{K}_{i}}w_{k}$ for $\Mathscr{K}_{i}:=\{k:\pi_{k}=i\}\subseteq[K]$. $\Mathscr{K}_{i}$ is the subset of users who have chosen dimension $i$. Since the weights are i.i.d. random variables, the only source of dependence between $S_{i}$ and $S_{j}$ lies in the number of users who have chosen dimensions $i$ and $j$, respectively. These numbers are $K_{i}:=|\Mathscr{K}_{i}|$ and are not independent. Indeed, the vector $(K_{1},K_{2},\dotsc,K_{N})$ is distributed according to a Multinomial law with probabilities $(1/N,1/N,\dotsc,1/N)$. In particular, the MGF of $(K_{1},K_{2})$ is
\[ M_{K_{1},K_{2}}(t_{1},t_{2})=\bigg( \frac{1}{N}(e^{t_{1}}+e^{t_{2}}+(N-2)) \bigg)^{\!K}, \]
and tends in the large system limit to
\[ M_{K_{1},K_{2}}(t_{1},t_{2}) \to e^{\beta(e^{t_{1}}-1)}\cdot e^{\beta(e^{t_{2}}-1)}, \]
where each term can be recognized as the MGF of a Poisson random variable with mean $\beta$. Since $K_{1}\perp K_{2}$ asymptotically, also $S_{1}\perp S_{2}$ from the independence of the weights.

\section{Verifying the Carleman Condition}
\label{app:Carleman}
Carleman's condition is $\sum_{k\geq 1} \bar{m}_{2k}^{-1/(2k)}=\infty$. We start off by upper bounding $\bar{m}_{2k}$ as follows:
\begin{equation}
\begin{aligned}
\bar{m}_{2k}&= \sum_{l=1}^{2k} \lah{2k}{l}\beta^l \\
& \overset{\textup{(a)}}{<} \sum_{l=1}^{2k} (2k-1)^{2k-l} \binom{2k}{l} \beta^l \\
& \overset{\textup{(b)}}{\leq} (2k-1)^{2k-1} \sum_{l=1}^{2k} \binom{2k}{l} \beta^l \\
& \overset{\textup{(c)}}{<} (2k-1)^{2k} (1+\beta)^{2k},\\
\end{aligned}
\end{equation}
where (a) follows from the inequality $\lahinline{2k}{l}=(2k-1)!/(l-1)!=(2k-1)(2k-2)\ldots (2k-(2k-l))<(2k-1)^{2k-l}$, $(\textup{b})$ derives from upper bounding $(2k-1)^{2k-l}$ with $(2k-1)^{2k-1}$, (c) is from the binomial formula by including in the sum the $l=0$ term. Therefore, $\bar{m}_{2k}^{1/(2k)}<(2k-1)(1+\beta)$, thus 
\[ \sum_{k\geq 1} \bar{m}_{2k}^{-1/(2k)}\geq \frac{1}{1+\beta} \sum_{k\geq 1} \frac{1}{2k-1} = \infty. \]

\section{Proof of \eqref{eq:etaminOPT}}\label{app:etaminOPT}
It is convenient to represent $C^{\textup{opt}}(\beta,\gamma)$ (in nats) as
\begin{equation}
\label{eq:COPTbis}
\hspace{-0.5ex}C^{\textup{opt}}(\beta,\gamma)\hspace{-0.25ex}=\hspace{-0.25ex} \sum_{k\geq 1} \frac{e^{-\beta} \beta^{k}}{k!}\! \int_{0}^{\gamma}\hspace{-0.25ex} k \exp(1/x)\,E_{1+k}(1/x) \frac{\diff{x}}{x},
\end{equation}
which can be derived by differentiating in \eqref{eq:COPT} under the integral sign with respect to $\gamma$ and integrating back after the integration with respect to $\lambda$. From the fundamental theorem of calculus, we have
\begin{equation*}
\frac{\partial}{\partial\gamma} \int_{0}^{\gamma} \exp(1/x)\,E_{1+k}(1/x) \frac{\diff{x}}{x} = \exp(1/\gamma)\, E_{1+k}(1/\gamma) \frac{1}{\gamma},
\end{equation*}
which tends to $1$ as $\gamma\to0$, hence, by L'H\^{o}pital's rule,
\begin{equation}\label{eq:1stderivative}
\hspace{-1ex}\lim_{\gamma\to0} \frac{C^{\textup{opt}}(\beta,\gamma)}{\gamma} \hspace{-0.25ex}=\hspace{-0.25ex} \lim_{\gamma\to0} \frac{\partial C^{\textup{opt}}(\beta,\gamma)}{\partial\gamma} \hspace{-0.25ex}=\hspace{-0.5ex} \sum_{k\geq 1} \frac{e^{-\beta} \beta^{k}}{k!} k \hspace{-0.25ex}=\hspace{-0.25ex}\beta.
\end{equation}

\section{Proof of \eqref{eq:slope0OPT}}\label{app:slope0OPT}
The second derivative of $C^{\textup{opt}}(\beta,\gamma)$ can be computed similarly to Appendix~\ref{app:etaminOPT}, which results in
\begin{multline}\label{eq:minus2ndderivative}
-\lim_{\gamma\to0}\frac{\partial^{2}}{\partial\gamma^{2}}C^{\textup{opt}}(\beta,\gamma) = -\sum_{k\geq 1} \frac{e^{-\beta} \beta^{k}}{k!} \lim_{\gamma\to0}\frac{\partial}{\partial\gamma} \exp(1/\gamma)\, E_{1+k}(1/\gamma) \frac{1}{\gamma}  \\= \sum_{k\geq 1} \frac{e^{-\beta} \beta^{k}}{k!} k(1+k) = 2\beta+\beta^{2}.
\end{multline}
The result follows by \eqref{eq:minus2ndderivative} and \eqref{eq:1stderivative}.

\section{Proof of \eqref{eq:slopeinfOPT}}\label{app:slopeinfOPT}
Using \eqref{eq:COPTbis} and the fundamental theorem of calculus yields
\begin{equation}
\hspace{-2ex}\lim_{\gamma\to\infty} \gamma\frac{\partial C^{\textup{opt}}_{\textup{lds}}}{\partial\gamma}
= \sum_{k\geq 1} \frac{e^{-\beta} \beta^{k}}{k!} \lim_{\gamma\to\infty} k \exp(1/\gamma)\, E_{1+k}(1/\gamma) 
= \sum_{k\geq 1} \frac{e^{-\beta} \beta^{k}}{k!} = 1-e^{-\beta}.
\end{equation}

\bibliographystyle{IEEEtran}
\bibliography{IEEEabrv,IEEE-conf-abrv,Biblio}

\begin{thebibliography}{10}
\providecommand{\url}[1]{#1}
\csname url@samestyle\endcsname
\providecommand{\newblock}{\relax}
\providecommand{\bibinfo}[2]{#2}
\providecommand{\BIBentrySTDinterwordspacing}{\spaceskip=0pt\relax}
\providecommand{\BIBentryALTinterwordstretchfactor}{4}
\providecommand{\BIBentryALTinterwordspacing}{\spaceskip=\fontdimen2\font plus
\BIBentryALTinterwordstretchfactor\fontdimen3\font minus
  \fontdimen4\font\relax}
\providecommand{\BIBforeignlanguage}[2]{{%
\expandafter\ifx\csname l@#1\endcsname\relax
\typeout{** WARNING: IEEEtran.bst: No hyphenation pattern has been}%
\typeout{** loaded for the language `#1'. Using the pattern for}%
\typeout{** the default language instead.}%
\else
\language=\csname l@#1\endcsname
\fi
#2}}
\providecommand{\BIBdecl}{\relax}
\BIBdecl

\bibitem{AndBuzChoetal:2014}
J.~G. Andrews, S.~Buzzi, W.~Choi, S.~V. Hanly, A.~Lozano, A.~C.~K. Soong, and
  J.~C. Zhang, ``{What Will 5G Be?}'' \emph{{IEEE} J. Sel. Areas Commun.},
  vol.~32, no.~6, pp. 1065--1082, June 2014.

\bibitem{DaiWanYuaetal:2015}
L.~Dai, B.~Wang, Y.~Yuan, S.~Han, C.~l. I, and Z.~Wang, ``Non-orthogonal
  multiple access for {5G}: solutions, challenges, opportunities, and future
  research trends,'' \emph{{IEEE} J. Sel. Areas Commun.}, vol.~53, no.~9, pp.
  74--81, Sept. 2015.

\bibitem{BocHeaLozetal:2014}
F.~Boccardi, R.~W. Heath, A.~Lozano, T.~L. Marzetta, and P.~Popovski, ``Five
  disruptive technology directions for {5G},'' \emph{{IEEE} Commun. Mag.},
  vol.~52, no.~2, pp. 74--80, Feb. 2014.

\bibitem{Qualcomm:2015}
``{5G Waveform \& Multiple Access Techniques},'' Qualcomm~Technologies,~Inc.,
  Nov. 2015.

\bibitem{CovTho:2012}
T.~Cover and J.~Thomas, \emph{Elements of information theory}, 2nd~ed.\hskip
  1em plus 0.5em minus 0.4em\relax New York, USA: John Wiley \& Sons Inc.,
  2012.

\bibitem{Choi:2016}
J.~Choi, ``{On HARQ-IR for Downlink NOMA Systems},'' \emph{{IEEE} Trans.
  Commun.}, vol.~64, no.~8, pp. 3576--3584, Aug. 2016.

\bibitem{Van:2012}
S.~Vanka, S.~Srinivasa, Z.~Gong, P.~Vizi, K.~Stamatiou, and M.~Haenggi,
  ``{Superposition Coding Strategies: Design and Experimental Evaluation},''
  \emph{{IEEE} Trans. Commun.}, vol.~11, no.~7, pp. 2628--2639, July 2012.

\bibitem{HosWatTaf:2008}
R.~Hoshyar, F.~Wathan, and R.~Tafazolli, ``{Novel Low-Density Signature for
  Synchronous CDMA Systems over AWGN channel},'' \emph{{IEEE} Trans. Signal
  Process.}, vol.~56, no.~4, pp. 1616--1626, Apr. 2008.

\bibitem{RazHosIm:2011}
R.~Razavi, R.~Hoshyar, M.~Imran, and Y.~Wang, ``{Information theoretic analysis
  of LDS scheme},'' \emph{{IEEE} Commun. Lett.}, vol.~15, no.~8, pp. 798--800,
  Aug. 2011.

\bibitem{BeePop:2009}
J.~Van De~Beek and B.~Popovic, ``{Multiple Access with Low-Density
  Signatures},'' in \emph{Proc. IEEE Glob. Telecommun. Conf. (GLOBECOM)},
  Honolulu, HI, 2009, pp. 1--6.

\bibitem{RazImaIm:2012}
R.~Razavi, A.~Mohammed, M.~Imran, R.~Hoshyar, and D.~Chen, ``{On receiver
  design for uplink low density signature OFDM (LDS-OFDM)},'' \emph{{IEEE}
  Trans. Commun.}, vol.~60, no.~11, pp. 3499--3508, Nov. 2012.

\bibitem{ZhaZhoZho:2017}
M.~Zhao, S.~Zhou, W.~Zhou, and J.~Zhu, ``{An Improved Uplink Sparse Coded
  Multiple Access},'' \emph{{IEEE} Commun. Lett.}, vol.~21, no.~1, pp.
  176--179, Jan. 2017.

\bibitem{CheRenGaoetal:2017}
S.~Chen, B.~Ren, Q.~Gao, S.~Kang, S.~Sun, and K.~Niu, ``{Pattern Division
  Multiple Access- A Novel Nonorthogonal Multiple Access for Fifth-Generation
  Radio Networks},'' \emph{{IEEE} Trans. Veh. Technol.}, vol.~66, no.~4, pp.
  3185--3196, Apr. 2017.

\bibitem{YuaYuLi:2016}
Z.~Yuan, G.~Yu, W.~Li, Y.~Yuan, X.~Wang, and J.~Xu, ``Multi-user shared access
  for internet of things,'' in \emph{Proc. IEEE Veh. Technol. Conf.
  (VTC-Spring)}, Nanjing, 2016, pp. 1--5.

\bibitem{VerSha:1999}
S.~Verd{\'u} and S.~Shamai, ``Spectral efficiency of {CDMA} with random
  spreading,'' \emph{{IEEE} Trans. Inf. Theory}, vol.~45, no.~2, pp. 622--640,
  Mar. 1999.

\bibitem{FerDiB:2015}
G.~C. Ferrante and M.-G. Di~Benedetto, ``Spectral efficiency of random
  time-hopping {CDMA},'' \emph{{IEEE} Trans. Inf. Theory}, vol.~61, no.~12, pp.
  6643--6662, Dec. 2015.

\bibitem{TseHan:1999}
D.~N.~C. Tse and S.~V. Hanly, ``Linear multiuser receivers: Effective
  interference, effective bandwidth and user capacity,'' \emph{{IEEE} Trans.
  Inf. Theory}, vol.~45, no.~2, pp. 641--657, Mar. 1999.

\bibitem{TseZei:2000}
D.~N.~C. Tse and O.~Zeitouni, ``Linear multiuser receivers in random
  environments,'' \emph{{IEEE} Trans. Inf. Theory}, vol.~46, no.~1, pp.
  171--188, Jan. 2000.

\bibitem{ShaVer:2001}
S.~Shamai and S.~Verd{\'u}, ``The impact of frequency-flat fading on the
  spectral efficiency of {CDMA},'' \emph{{IEEE} Trans. Inf. Theory}, vol.~47,
  no.~4, pp. 1302--1327, May 2001.

\bibitem{MonTse:2006}
A.~Montanari and D.~N.~C. Tse, ``Analysis of belief propagation for non-linear
  problems: The example of {CDMA} (or: How to prove tanaka's formula),'' in
  \emph{Proc. IEEE Inf. Theory Workshop (ITW)}, Punta del Este, 2006, pp.
  160--164.

\bibitem{RaySaad:2007}
J.~Raymond and D.~Saad, ``Sparsely spread {CDMA}- a statistical mechanics-based
  analysis,'' \emph{J. Phys. A: Math. Theor.}, vol.~40, no.~41, pp.
  12\,315--12\,333, 2007.

\bibitem{YosTan:2006}
M.~Yoshida and T.~Tanaka, ``Analysis of sparsely-spread {CDMA} via statistical
  mechanics,'' in \emph{Proc. IEEE Int. Symp. Inf. Theory (ISIT)}, 2006, pp.
  2378--2382.

\bibitem{Tan:2002}
T.~Tanaka, ``A statistical-mechanics approach to large-system analysis of
  {CDMA} multiuser detectors,'' \emph{{IEEE} Trans. Inf. Theory}, vol.~48,
  no.~11, pp. 2888--2910, 2002.

\bibitem{Ver:2002}
S.~Verd{\'u}, ``Spectral efficiency in the wideband regime,'' \emph{{IEEE}
  Trans. Inf. Theory}, vol.~48, no.~6, pp. 1319--1343, June 2002.

\bibitem{Ver:1986}
------, ``Capacity region of {Gaussian} {CDMA} channels: The symbol-synchronous
  case,'' in \emph{Proc. Allerton Conf. on Commun., Control, and Computing
  (Allerton)}, Oct. 1986, pp. 1025--1034.

\bibitem{Gir:1990}
V.~L. Girko, \emph{Theory of Random Determinants}.\hskip 1em plus 0.5em minus
  0.4em\relax Dordrecht, Netherlands: Kluwer Academic, 1990.

\bibitem{BaiSil:2010}
Z.~Bai and J.~W. Silverstein, \emph{Spectral Analysis of Large Dimensional
  Random Matrices}, 2nd~ed.\hskip 1em plus 0.5em minus 0.4em\relax New York,
  USA: Springer, 2010, vol.~20.

\bibitem{Fel2:1968}
W.~Feller, \emph{An introduction to probability theory and its
  applications}.\hskip 1em plus 0.5em minus 0.4em\relax New York, USA: John
  Wiley \& Sons Inc.,, 1968, vol.~2.

\bibitem{Com:1974}
L.~Comtet, \emph{{Advanced Combinatorics: The Art of Finite and Infinite
  Expansions (Revised and Enlarged Edition)}}.\hskip 1em plus 0.5em minus
  0.4em\relax Dordrecht, Holland and Boston, USA: D.~Reidel Publishing Co.,
  1974.

\bibitem{GuoVer:2005}
D.~Guo and S.~Verd\'u, ``{Randomly Spread CDMA: Asymptotics via Statistical
  Physics},'' \emph{{IEEE} Trans. Inf. Theory}, vol.~51, no.~6, pp. 1983--2010,
  June 2005.

\bibitem{AbrSte:1972}
M.~Abramowitz and I.~A. Stegun, \emph{Handbook of Mathematical Functions with
  Formulas, Graphs, and Mathematical Tables}.\hskip 1em plus 0.5em minus
  0.4em\relax New York, USA: Dover Publications, 1972.

\end{thebibliography}
\end{document}